\documentclass[11pt, a4paper]{article}

\usepackage{fullpage}
\usepackage{amsthm,amsmath,amssymb,amsfonts}
\usepackage{url,hyperref}

\usepackage{color}

\newcommand{\blue}[1]{#1}

\newcommand{\ignore}[1]{}

\newcommand{\ii}{\mathsf{i}}
\newcommand{\ee}{\mathsf{e}}
\newcommand{\eps}{\varepsilon}

\def\makeletter#1{%
\expandafter \newcommand \csname b#1\endcsname {\mathbb{#1}}%
\expandafter \newcommand \csname c#1\endcsname {\mathcal{#1}}%
\expandafter \newcommand \csname t#1\endcsname {\widetilde{#1}}%
\expandafter \newcommand \csname ct#1\endcsname {\widetilde{\mathcal{#1}}}%
}
\def\makeletters(#1#2){\makeletter#1\ifx#2.\else\makeletters(#2)\fi}
\makeletters(QWERTYUIOPASDFGHJKLZXCVBNM.)
\def\makeletter#1{%
\expandafter \newcommand \csname t#1\endcsname {\widetilde{\csname #1\endcsname}}%
}
\def\makeletters(#1,#2){\makeletter{#1}\ifx#2.\else\makeletters(#2)\fi}
\makeletters(Gamma,Delta,Theta,Lambda,Xi,Pi,Sigma,Upsilon,Phi,Psi,Omega,.)

\def\makeSkob#1#2#3{%
\def\LLL{\mathopen{}\mathclose\bgroup\left} \def\RRR{\aftergroup\egroup\right}
\expandafter \edef \csname #1\endcsname #2##1#3{\SkobInner}
\def\LLL{} \def\RRR{}
\expandafter \edef \csname #1O\endcsname #2##1#3{\SkobInner}
\def\LLL{\bigl} \def\RRR{\bigr}
\expandafter \edef \csname #1A\endcsname #2##1#3{\SkobInner}
\def\LLL{\Bigl} \def\RRR{\Bigr}
\expandafter \edef \csname #1B\endcsname #2##1#3{\SkobInner}
\def\LLL{\biggl} \def\RRR{\biggr}
\expandafter \edef \csname #1C\endcsname #2##1#3{\SkobInner}
\def\LLL{\Biggl} \def\RRR{\Biggr}
\expandafter \edef \csname #1D\endcsname #2##1#3{\SkobInner}
}

\def\SkobInner{\LLL(##1\RRR)} \makeSkob{s}[]
\def\SkobInner{\LLL[##1\RRR]} \makeSkob{sk}[]
\def\SkobInner{\LLL\lbrace##1\RRR\rbrace} \makeSkob{sfig}{}{}
\def\SkobInner{\LLL\lfloor##1\RRR\rfloor} \makeSkob{floor}[]
\def\SkobInner{\LLL\lceil##1\RRR\rceil} \makeSkob{ceil}[]
\def\SkobInner{\LLL\langle##1\RRR\rangle} \makeSkob{ip}<>
\def\SkobInner{\LLL\lvert##1\RRR\rangle} \makeSkob{ket}|>
\def\SkobInner{\LLL\lvert##1\RRR\rvert} \makeSkob{abs}||
\def\SkobInner{\LLL\lVert##1\RRR\rVert} \makeSkob{norm}||
\def\SkobInner{\LLL\lVert##1\RRR\rVert_{\noexpand\mathrm F}} \makeSkob{normFrob}||
\def\SkobInner{\LLL\lVert##1\RRR\rVert_{\noexpand\mathrm{tr}}} \makeSkob{normtr}||

\def \elem[#1]{[\![#1]\!]}
\def \bigfrac#1/{\left.#1\right/}
\def \bigfracR/#1.{\left/#1\right.}

\newcommand{\pfstart}{\begin{proof}} 
\newcommand{\pfsketch}{\begin{proof}[Proof sketch]}
\newcommand{\pfend}{\end{proof}}

\newcommand{\itemstart}{\begin{itemize}\itemsep0pt}
\newcommand{\itemend}{\end{itemize}}
\newcommand{\descrstart}{\begin{description}\itemsep0pt}
\newcommand{\descrend}{\end{description}}
\newcommand{\enumstart}{\begin{enumerate}\itemsep0pt}
\newcommand{\enumalpha}{\begin{enumerate}[(a)] \itemsep0pt}
\newcommand{\enumend}{\end{enumerate}}

\renewcommand{\*}{^{\phantom{*}}}

\newcommand{\myhref}[2]{\url{#1}}  

\newcommand{\textsfup}[1]{\textsf{\textup{#1}}}

\newcommand{\twoxor}{\textsfup{\mbox{2-XOR}}}

\newcommand{\ksum}{\textsfup{\mbox{$k$-SUM}}}
\newcommand{\twosum}{\textsfup{\mbox{2-SUM}}}
\newcommand{\hadamard}{\circ}
\newcommand{\comp}{\circ}
\newcommand{\sampledfrom}{\sim}

\newcommand{\Adv}{\overline{\textsfup{Adv}}}

\newcommand{\I}{\textup{I}}    	
\newcommand{\one}{\textsfup{J}}	

\newcommand{\X}{\textsfup{X}}
\newcommand{\Y}{\textsfup{Y}}

\newcommand{\ED}{\textsfup{ED}}
\newcommand{\SED}{\textsfup{SED}}

\newcommand{\HH}{\textsfup{H}}

\newcommand{\psearch}{\textsfup{pSEARCH}}

\newcommand{\F}{\textsfup{F}}

\newcommand{\G}{\textsfup{G}}
\newcommand{\Gr}{\mathbb G}

\newcommand{\Image}{\textsfup{Im}}

\newcommand{\outern}{n}
\newcommand{\innern}{\ell}

\newcommand{\epsn}{\nu}
\newcommand{\epsN}{\nu}

\newtheorem{theorem}{Theorem} 

\newtheorem{definition}{Definition}

\newtheorem{lemma}{Lemma}
\newtheorem{clm}{Claim}

{\itshape}{\normalfont}

\usepackage{authblk}

\def\mycommand#1#2{
\expandafter\newcommand \csname#1\endcsname {#2}%
}

\begin{document}

\title{Provably secure key establishment against quantum adversaries}

\author[1]{Aleksandrs Belovs}
\author[2,3]{Gilles Brassard}
\author[4]{Peter H{\o}yer}
\author[5]{Marc Kaplan}
\author[6]{Sophie Laplante}
\author[2]{Louis Salvail}
\affil[1]{University of Latvia\\
  \texttt{stiboh@gmail.com}}
\affil[2]{DIRO, Universit\'e de Montr\'eal\\
  \texttt{\{brassard,\,salvail\}@iro.umontreal.ca}}
\affil[3]{Canadian Institute for Advanced Research}
\affil[4]{Department of Computer Science, University of Calgary\\
  \texttt{hoyer@ucalgary.ca}}
\affil[5]{School of Informatics, University of Edinburgh\\
  \texttt{kapmarc@gmail.com}}
\affil[6]{IRIF, Universit\'e Paris Diderot\\
  \texttt{laplante@irif.fr}}

\maketitle

\begin{abstract}

At \textsc{Crypto}~2011, some of us had proposed a
family of cryptographic protocols for key establishment capable of
protecting quantum \emph{and classical} legitimate parties unconditionally against a \emph{quantum}
eavesdropper in the query complexity model. 
Unfortunately, our security proofs were \blue{unsatisfactory}
from a cryptographically meaningful perspective because they were sound only in a
worst-case scenario. 
Here, we extend our results and prove that for any \mbox{$\eps>0$}, there is a classical protocol that
allows the legitimate parties to establish a common key after $O(N)$ expected queries to a random
oracle, yet any quantum eavesdropper will have a vanishing probability of learning their key
after $O(N^{1.5-\eps})$ queries to the same oracle.
The~vanishing probability applies to a typical run of the protocol.
If~we allow the legitimate parties to use a quantum computer as well, their advantage
over the quantum eavesdropper \mbox{becomes} arbitrarily close to the quadratic advantage
that classical legitimate parties enjoyed over classical eavesdroppers in the seminal 1974
work of Ralph Merkle.
\blue{Along the way, we develop new tools to give lower bounds on the number of quantum queries required to distinguish two probability distributions. This method in itself could have multiple applications in cryptography. We~use it here to study average-case quantum query complexity, for which we develop a new composition theorem of independent interest.}
\end{abstract}

\section{Introduction}

Not taking classified work within secret services into consideration~\cite{Cocks},
Ralph Merkle is the first person to have asked---and solved---the question of
secure communications over insecure channels~\cite{merkle74}.
In~his seminal (rejected!)\ 1974 project for a Computer Security course
at the University of California, Berkeley, he discovered that it is possible for two
people who want to communicate securely to establish a secret key
by communicating over an authenticated channel that provides no protection
against eavesdropping. Merkle's solution to this conundrum offers
\emph{quadratic security} in the sense that if the legitimate parties---codenamed
Alice and Bob---are willing to expend an effort in the order of $N$, for some security
parameter~$N$, they can establish a key that no eavesdropper---codenamed Eve---can
discover with better than vanishing probability without expending an effort in the order of~$N^2$.

This quadratic security may seem unattractive compared to the potential exponential security
entailed by the subsequently discovered key establishment protocols of Diffie and
Hellman~\cite{DH} and Rivest, Shamir and Adleman~\cite{RSA}, to name a few.
However, the security of those currently ubiquitous cryptographic solutions will 
be compromised with the advent of full-scale quantum computers,
as discovered by Peter Shor more than two decades ago~\cite{Shor}.
And~even if a quantum computer is never built,
no~one has been able to prove their security against classical attacks,
nor that of quantum-resistant candidates based, for instance, on short vectors in lattices.
Furthermore, Merkle had already understood in 1974 that quadratic security
\emph{could} be practical if the underlying one-way function (see~below) can be computed very
quickly: if~it takes one nanosecond to compute the function and legitimate users
are willing to spend one second each, a classical adversary who could only invert
the function by exhaustive search would require fifteen expected \emph{years} to break
Merkle's original scheme.

The main interest of Merkle's 
solution is that it offers \emph{provable} security,
at least in the \emph{query model} of computational complexity, 
a model closely related to the random oracle model.
In~this model, we
assume the existence of a \emph{black-box} function \mbox{$f:D \rightarrow R$}
from some domain $D$ to some range~$R$, so that the only way to learn something about
this function is to query the value of $f(x)$ on inputs \mbox{$x \in D$} that can be
chosen arbitrarily. The~\emph{query complexity} of some problem given $f$ is
defined as the expected number of calls to~$f$ required to solve the problem,
using the best possible algorithm. In~our case of interest, we shall consider \emph{random}
black-box functions, meaning that for each \mbox{$x \in D$}, the value of~$f(x)$
is chosen uniformly at random within~$R$, independently of the value of $f(x')$
for any other~\mbox{$x' \in D$}. Provided the size $r$ of $R$ is sufficiently large compared
to the size $d$ of~$D$, such a random function is automatically one-to-one, except
with vanishing probability. The main characteristic of these black-box random functions
that is relevant to the proof of security of Merkle's scheme is that, given a randomly chosen point $y$
in the image of $f$, the only (classical) approach to finding an $x$ so that \mbox{$f(x)=y$}
is exhaustive search: we~have to try $x$'s one after another until a solution is found.
Indeed, whenever we try some $x'$ and find that \mbox{$f(x') \neq y$}, the \emph{only} thing
we have learned is that this particular $x'$ is not a solution. Provided the function is
indeed one-to-one, we expect to have to query the function
$d/2$ times on 
average in order to find the unique solution.

One~may argue that black-box random functions \blue{do not} exist in real life,
but we can \mbox{replace} them in practice with
one-way functions---provided \blue{\emph{they}} exist!---which is what Merkle meant by
``one-way encryption'' in his 1974 class assignment~\cite{merkle74}.
Thus, we can base the \mbox{security} of Merkle's scheme on the \emph{generic} assumption
that one-way functions exist, which is \mbox{unlikely} to be broken by a quantum computer,
rather than the assumption that \emph{specific} computational problems such as factorization
or finding short vectors in lattices are difficult, at least the first one of which is
known not to hold on a quantum computer.
Can~we do better than provable \emph{quadratic} security in the query model?
This question remained open for 35 years, and was finally settled in the negative by
Boaz Barak and Mohammad Mahmoody-Ghidary~\cite{BarMah09}, building on earlier
work of Russell Impagliazzo and Steven Rudich~\cite{ImpagliazzoRudich}: any protocol by
which the legitimate parties can obtain a shared key after $O(N)$ expected queries to a
black-box random function can be broken with $O(N^2)$ expected queries to the same
black~box.

It was apparently noticed for the first time by one of us in 2005,
and published a few years later~\cite{ICQNM},
that Merkle's original 1974 scheme~\cite{merkle74}, as well as his
better known subsequently published \emph{puzzles}~\cite{merkle78},
are broken by Grover's algorithm~\cite{grover} on a quantum computer.
This attack assumes that the eavesdropper can query the function
in quantum superposition, which is perhaps not reasonable if the function
is 
provided as a \emph{physical} classical black box, but is completely reasonable if it is given
by the publicly-available \emph{code} of a one-way function
(as~originally envisioned by Merkle).
If~the legitimate parties are also endowed with a quantum computer,
the same paper~\cite{ICQNM} gave an obvious fix, by which the
legitimate parties can
establish a key after $O(N)$ quantum queries to the black-box function, but
no quantum eavesdropper can discover it with better than vanishing probability
without querying the function $O(N^{3/2})$ times.
That paper made the explicit conjecture that this was best possible when
quantum codemakers are facing quantum codebreakers in the game of provable
security in the random black-box model.
The~issue of protecting classical codemakers
against quantum codebreakers was not addressed in Ref.~\cite{ICQNM}.

At the \textsc{Crypto}~2011 conference~\cite{bhkkls:crypto}, several of us
disproved the conjecture of Ref.~\cite{ICQNM}
with the introduction of a new quantum protocol that no quantum
eavesdropper could break without querying the black-box functions $\blue{\Omega}(N^{5/3})$ times.%
\footnote{\,The word ``function\textbf{s}'' is plural because the 2011 protocol required \emph{two} black-box random functions.}
We~also offered the first protocol provably capable of protecting \emph{classical}
codemakers against \emph{quantum} codebreakers,
although $O(N^{13/12})$ queries in superposition sufficed for the quantum eavesdropper
to obtain the not-so-secret key.
Unfortunately, our security proofs 
were worked out in the traditional computational complexity \emph{worst-case} scenario.
In~other words, it was only proved that any quantum eavesdropper
limited to $o(N^{5/3})$ or $o(N^{13/12})$ queries, depending on whether the legitimate
parties are quantum or classical, 
would be likely to fail \emph{on at least one possible
instance} of the protocol. This did not preclude that most instances of the protocol
could result in insecure keys against an eavesdropper who would work no harder
than the legitimate parties. Said otherwise, our \textsc{Crypto}~2011 result was of 
limited cryptographic significance.

In subsequent work~\cite{bhkkls:arxivv2}, we 
claimed to have provided a proper average-case
analysis of our protocols, rendering them cryptographically meaningful,
so~that any quantum eavesdropper has a vanishing probability of learning the key
after only $o(N^{5/3})$ or $o(N^{7/6})$ queries\,%
\footnote{\,For classical legitimate parties, the~$o(N^{13/12})$ of Ref.~\cite{bhkkls:crypto}
had been improved to~$o(N^{7/6})$ in Ref.~\cite{bhkkls:arxivv2}.},
where the probabilities are taken not only over the execution of the eavesdropping
algorithm but also over the instance of the protocol run by the legitimate parties.
We~also extended our results to two sequences of protocols based on the \ksum{}
problem (Definition~\ref{def:ksum} in Section~\ref{sec:protocols}),
where \mbox{$k \ge 2$} is an integer parameter,
in which the legitimate parties query the black-box
random functions $O(kN)$ times. It~was claimed~that any quantum eavesdropper had
a vanishing probability of learning the key after \smash{$o(N^{\frac12+\frac{k}{k+1}})$}
or \smash{$o(N^{1+\frac{k}{k+1}})$} queries, against the classical or the quantum protocol
parametrized by~$k$,
respectively. Again, this was claimed to hold not only in the
cryptographically-challenged worst-case scenario, but also when the probabilities are taken
over the protocols being run by the legitimate parties.

Unfortunately, all our average-case analyses in Ref.~\cite{bhkkls:arxivv2} were incorrect!
The case \mbox{$k=2$} can be fixed rather easily, 
hence the insufficiency of $o(N^{5/3})$ queries for a quantum-against-quantum protocol
and of $o(N^{7/6})$ queries for a classical-against-quantum protocol in a cryptographically
significant setting can be derived from the incorrect arguments provided in Ref.~\cite{bhkkls:arxivv2}.
However, we also claimed in Ref.~\cite{bhkkls:arxivv2} that the case
\mbox{$k>2$} could be proved in ways ``similar to'' when \mbox{$k=2$}.
This was a mistake due to a fundamental difference in the \ksum{} problem
whether \mbox{$k=2$} or~\mbox{$k>2$}.
Whereas the \twosum{} problem is easily seen to be random self-reducible,
so that its hardness in worst case implies its hardness on average, this does not seem
to be the case for the \ksum{} problem when \mbox{$k>2$}.
In~particular, the worst-case lower bound proved by
Aleksandrs Belovs and Robert \v{S}palek~\cite{spalek:kSumLower}
on the difficulty of solving the \ksum{} problem on a quantum computer does not extend in
any obvious way to a lower bound on 
average. And~without such an average lower bound,
our results claimed in Ref.~\cite{bhkkls:arxivv2} go up in smoke for \mbox{$k>2$}.
Furthermore, for a technical reason explained later, 
even such an average lower bound would not suffice.

In this paper, we overcome all these problems and give a correct and cryptographically
meaningful\,\footnote{\,To~be honest, it is not entirely cryptographically meaningful to restrict
the analysis to the number of calls to the black-box functions, taking no account of
the computing time that may be required outside those calls.
However, if we also restrict the legitimate expected \emph{time} to be in~$O(N)$,
then our quantum protocol with \mbox{$k=3$} remains valid and provably resists any
\smash{$o(N^{7/4})$}-time quantum eavesdropping attack, which was claimed in
Ref~\cite{bhkkls:arxivv2}, but with a fundamentally incorrect proof.} security proof for all our
protocols from Ref.~\cite{bhkkls:arxivv2}.
Consequently, we~prove that for any $\eps>0$ there is a classical protocol that
allows the legitimate parties to establish a common key after $O(N)$ expected queries to
black-box random functions, yet any quantum eavesdropper will have a vanishing probability
of learning their key after $O(N^{1.5-\eps})$ queries to the same oracle.
The~vanishing probability is over the randomness in the actual run of the protocol followed
by that of the eavesdropper's algorithm. If~we allow the legitimate parties to use quantum
computers as well, their advantage over the quantum eavesdropper becomes arbitrarily close
to the quadratic advantage that classical legitimate parties enjoyed over classical eavesdroppers
in the seminal 1974 work of Ralph Merkle~\cite{merkle74}.

Our~results require new tools in quantum query complexity,
which are of independent \mbox{interest}.
\blue{In~particular, we introduce techniques to lower-bound the quantum query complexity
of distinguishing between two probability distributions, which we use to extend the
adversary lower bound method in \mbox{order} to handle average-case complexity,
but they could have other uses in cryptography.
This approach is necessary for the distributions of inputs considered here
because the associated decision problems become trivial on average, which}
prevents \blue{us from} applying the average-case method developed
in Ref.~\cite{belovs:onThePower}.
\blue{Furthermore}, we prove a composition theorem for this new lower bound method,
\blue{extending that of} Ref.~\cite{bhkkls:crypto},
which was valid only to prove cryptographically irrelevant worst-case lower bounds. 
Using these two tools, we prove that any quantum eavesdropper who does not make
a prohibitive number of calls to the black-box functions will fail to break a typical
instance of the protocol, except with vanishing probability.  

This work fits in the general framework of ``Cryptography in a quantum world''~\cite{Harachov},
which addresses the question:
``Is~the fact that we live in a quantum world a blessing or a curse for codemakers?''.
It~is a blessing if we allow quantum communication, thanks to Quantum Key Establishment
(aka~Quantum Key Distribution---QKD)~\cite{BB84}, at least if the protocols
can be imple\-mented faithfully according to theory~\cite{HKL,Makarov}.
On~the other hand, it is a curse if we 
continue to use the current cryptographic infrastructure,
which pretends to secure the Inter\-net at the risk of falling prey
to upcoming quantum computers.
However, it is mostly a draw in the realm of provable query complexity in the black-box model
considered in this paper
since codemakers enjoy a quadratic (or~arbitrarily close to being quadratic) advantage over
codebreakers in both an all-classical or an all-quantum world, at least in terms of query complexity (but~see footnote~\thefootnote\ again).
Furthermore, the known proof that quadratic security is best possible in an all-classical
world~\cite{BarMah09}
does not extend to the all-quantum world, and hence the (unlikely) possibility remains that
a more secure protocol could exist in our quantum world.

The rest of the paper is organized as follows. 
Section~\ref{sec:prelim} lists all the techniques and related notations that are used throughout the paper. 
Section~\ref{sec:protocols} recalls the classical and quantum protocols from Refs~\cite{bhkkls:crypto,bhkkls:arxivv2}.
In Section~\ref{sec:avgadv}, we introduce a new method to prove lower bounds on
\blue{the difficulty of distinguishing between two probability distributions,
which we use to study}
average-case quantum query complexity. This method extends the extensively studied adversary method.
We then apply this method to the \blue{$\ksum$ problem} in Section~\ref{sec:ksum}, which is at the heart of our hardness result.
Finally, in Section~\ref{sec:composition}, we prove a composition theorem for the new
adversary method introduced in Section~\ref{sec:avgadv}.
This allows us to conclude that typical runs of the protocols from
Refs~\cite{bhkkls:crypto,bhkkls:arxivv2} are indeed secure against quantum adversaries.

\section{Preliminaries and Notation}
\label{sec:prelim}

At the heart of this work is a lower bound on the quantum query complexity of a generalisation 
of the $\ksum$ problem.  Many techniques have been given to prove such lower bounds
in the worst-case scenario, including the adversary method~\cite{amb02,hoyer:advNegative,LMRSS11}. 
This method is based on the spectral norm of a matrix, $\Gamma$, indexed in the rows and
columns by inputs to the problem.  Roughly, each entry of the matrix $\Gamma[x,y]\in \mathbb{R}$
can be thought of as representing the hardness of distinguishing inputs $x$ and $y$.  
It~is known that for Boolean functions, the (negative) adversary bound is multiplicative 
under function composition~\cite{hoyer:advNegative}.  For non-Boolean functions, a general composition 
theorem fails to hold, as counterexamples can be found.
Nevertheless, it was shown in Ref.~\cite{bhkkls:crypto} that the adversary method \emph{is}
multiplicative under composition with (non-Boolean) unstructured search problems. 

In this paper, we extend the quantum adversary method to average-case complexity,
which is crucial for cryptographic applications,
and we show that a similar composition property holds for this measure.
As~for the adversary bound, this method is based on the spectral norm of
matrices, and involves probability distributions.
Below, we summarize the notation related to functions, algebra and probabilities, used throughout the paper.

We consider \emph{decision} or \emph{search} problems denoted $\F, \G$ or $\HH$. These problems are on abelian groups, which are denoted $\Gr$, or $\Gr_m$ when we want the order $m$ of the group to appear explicitly. The group operation is denoted ``$+$'' and its inverse ``$-$''. 
For a decision problem $\F$, the inputs in the language $\F$ are called \emph{positive} and the inputs not in the language are \emph{negative}.
We compose our problems with an unstructured search problem to make them harder. To do so, we need to add to the alphabet an element that does not belong to $\mathbb G$. We denote this element~``$\star$''.

Fix two problems $\F: A^\outern \rightarrow B$ and $\G: C \rightarrow A$ for some $\outern \in \mathbb N$. Then, the composed problem
$\F \comp \G^n: C^n \rightarrow B$ is defined by $\F \comp \G^n(x_1, \ldots x_n) = \F(\G(x_1), \ldots, \G(x_n))$ for $(x_1, \ldots x_n) \in C^n$.

For any positive integer $n$ we use $[n]$ to denote the set of $n$ elements
\mbox{$\{0,1,2,\ldots,n-1\}$}.
We~only make use of basic concepts of quantum computing: states, unitary operations and measurements.
These notions are used in Section~\ref{sec:avgadv}, but even there, the calculations boil down to basic linear algebra.
The entries of an $n \times m$ matrix $\Gamma$ are denoted $\Gamma[x,y]$,
where $x \in [n]$ 
and $y \in [m]$. 
For $X \subseteq [n]$ and $Y \subseteq [m]$, $\Gamma^{X,Y}$ is the restriction of $\Gamma$ to the rows and columns in $X$ and $Y\!$, respectively.

The direct sum of spaces, operators, matrices or vectors is denoted ``$\bigoplus$''. The inner product of two states 
(or vectors in an Hilbert space) $\psi$ and $\phi$ is $\ipA <\psi, \phi>$. 
For~a matrix $A$, we use $\norm |A|$ for its spectral norm, that is, its largest singular value, 
and $\normFrob |A|$ for the Frobenius norm, that is, the square root of the sum of the
\blue{squares of the moduli} of its elements.
For~two matrices $A$ and $B$, we denote $A \hadamard B$ their entrywise (or Hadamard) product.
We~make use of the two following matrices: the $n\times n$ identity matrix $\I_n$ and the $n\times n$ all-one matrix $\one_n$.

We use  $\cP$ and $\cQ$ for probability distributions over inputs to the problems. The \emph{support} of a distribution is the
set of elements with non-zero probability. We sometimes identify distributions with vectors.
More precisely, if $p_x$ is the probability of $x$ in $\cP$, we can consider the vector $\cP$ given by the entries $\cP[x] = p_x$.
We~use ``$X\sim \cP$\,'' to denote that the random variable $X$ is sampled from $\cP$. 
In~this case, it is the variable whose probability is given by $\Pr[X =x] = p_x$.
In~the specific case of sampling an element $x$ uniformly at random from a set~$D$, we use $x \in_R D$.
We~also use the indicator function $1_{x\neq y}$ whose value is 1 if $x\neq y$ and 0 otherwise.

We sometimes consider sequences of probabilities, such as the accepting probability $\epsn_n$ of an algorithm (for a decision problem) as a function of the input size~$n$. For simplicity, we often omit  the subscript $n$, in which case ``\,$\epsn$\,'' should be understood as a function of~$n$.
We~call such a sequence $\epsn$  \emph{vanishing} if $\epsn= o(1)$. If $\epsn$ decreases faster than \blue{the inverse of} any polynomial, we say that the event is \emph{negligible}.

\section{Provably Secure Key Establishment Protocols}
\label{sec:protocols}

With the exception of Merkle's more famous ``puzzles''~\cite{merkle78},
\blue{all} key establishment protocols based on black-box random functions
(which Merkle called ``one-way encryption'') begin in a way that is essentially
identical to Merkle's original 1974 idea~\cite{merkle74}, with possible inessential
differences\,\footnote{\,In~Merkle's original scheme, there is no asymmetry between
Alice and Bob, as they both ``guess at keywords'' and share and compare their one-way encryptions
until they discover that they have guessed at the same keyword.
In~all the protocols considered here, Alice goes first and Bob works from there.}.
Given a black-box random function \mbox{$f:D \rightarrow R$}
from some domain $D$ to some range~$R$, Alice chooses random elements
\mbox{$x_i \in_R D$} 
and she obtains \mbox{$y_i = f(x_i)$},
which she sends to Bob over an authenticated channel on which Eve can freely eavesdrop.
This defines the sets $X$ of $x_i$'s and $Y$ of $y_i$'s, of which $X$ is private
information kept by Alice whereas $Y$ becomes known to all parties, including Eve.
Upon receiving this information, Bob's first task is to find one or several preimage(s) under $f$ of
\emph{any} of the points sent by Alice. 

The various schemes that were considered in
Refs~\cite{merkle74,ICQNM,bhkkls:crypto,bhkkls:arxivv2}
differ in how Bob proceeds to find the preimage(s), how many such preimages
he needs to find, and how he informs Alice of which preimage(s) he has found.
In~Merkle's original scheme~\cite{merkle74}, he needs to find a single preimage.
This is done by querying~$f$ on random points in its domain until some $x$
is found such that \mbox{$f(x) = y \in Y\!$}. 
Afterwards, Bob sends $y$ back to Alice, who can find efficiently
the corresponding $x$ because it is among her set~$X$, which she had kept.
This shared $x$ \blue{becomes} their secret key.
The intuition behind the security of this scheme stems from the freedom in Bob's task
to invert $f$ on any element of $Y\!$,
compared to how stringent Eve's~is since she must invert it on the specific element
that Bob had inverted by chance.

To~be more precise, let $N$ be a safety parameter, let the domain of $f$ contain $N^2$
points and its range be \blue{of size} $N^5$, \blue{which} is large enough to ensure that $f$ is one-to-one except with
vanishing probability. If~Alice chooses $N$ random points in the
domain of~$f$ and Bob tries random such points as well until he hits upon an $x$
such that \mbox{$f(x) \in Y\!$}, it is easy to see that both Alice and Bob need query
function $f$ an expected number of $N$ times. However, a classical Eve
\blue{requires an expected} $N^2/2$ queries,
which gives a quadratic advantage to the legitimate parties.

Unfortunately, 
inverting one specific point in the image of $f$ with the help of a quantum computer
requires only \mbox{$\frac{\pi}{4} \sqrt{N^2} = \frac{\pi}{4} N$}
queries to~$f$ by way of Grover's algorithm~\cite{grover},
which is slightly \emph{fewer} than the effort required by the legitimate parties.
This is why Merkle's original scheme is totally broken against
a quantum adversary, as first pointed out in Ref.~\cite{ICQNM}.
In order to restore security, 
two main modifications to Merkle's original scheme have been considered,
as we now proceed to describe.

\subsection{Variations on Merkle's Idea}\label{sec:variations}

If~we require Bob to find $k$ distinct preimages among the $N$ points sent by Alice,
for some~\mbox{$k>1$}, rather than a single one, he will only have to work roughly
$k$ times as hard, provided \mbox{$k \ll N$}. The key shared by Alice and Bob could
then be the concatenation of those preimages in the order in which the corresponding
images were sent by Alice in the first step. But how can Bob tell Alice which preimages
he was able to find in a way that will force Eve to make much more queries than~her?
A~first solution was proposed in Ref.~\cite{bhkkls:crypto} for the case~\mbox{$k=2$},
but a much simpler one was given subsequently in Ref.~\cite{bhkkls:arxivv2} for
arbitrary~$k$.
The idea is to introduce a second black-box random function $t$ from the same domain
to some sufficiently large group~$\Gr$.
If~Bob finds preimages \mbox{$x_{i_1}, x_{i_2}, \ldots, x_{i_k} \in X$},
with \mbox{$1 \le i_1 < i_2 < \cdots < i_k \le N$},
and sends \mbox{$w=t(x_{i_1})+t(x_{i_2})+\cdots t(x_{i_k})$} to Alice,
she needs only call black-box function $t$ on the $N$ points she had kept in~$X$
in order to
determine Bob's $k$ preimages,
provided the order
of $\Gr$ was chosen sufficiently large to ensure the uniqueness of the
solution, except with vanishing probability. Taking the order to be $N^{4k+1}$
is sufficient to ensure this.
Furthermore, she can do this efficiently, in terms of computing time, when~\mbox{$k=2$}.
Hence, Alice needs to query each of functions $f$ and $t$ exactly $N$ times,
whereas Bob needs to query function $f$ an expected $O(kN)$ times and
function $t$ exactly $k$ times.

How~difficult is the cryptanalytic task for quantum Eve, who has seen the $y$'s sent
from Alice to Bob and the single $w$ sent from Bob to Alice?
We~gave an explicit algorithm based on quantum walks~\cite{MNRS} in Hamming graphs
in Ref.~\cite{bhkkls:arxivv2}, which allows her to discover the secret key after
$O(N^{1/2+k/(k+1)})$ calls to the black-box functions. 
In~the same paper, we claimed that a matching $\Omega(N^{1/2+k/(k+1)})$ lower bound holds
for a typical instance of the protocol, which is formally stated in Theorem~\ref{main-class-thm}
below, but the proof proposed in Ref.~\cite{bhkkls:arxivv2} fails for \mbox{$k>2$} in a way that cannot be repaired.
The main purpose of the present paper is to offer a correct proof of this theorem.
It~follows that for any fixed \mbox{$\eps>0$}, there is a \emph{classical} key establishment protocol (taking \mbox{$k=\lfloor 1/\eps\rfloor$}) that
allows the legitimate parties to establish a shared key after $O(N)$ expected queries to
black-box random functions $f$ and $t$, yet any \emph{quantum} eavesdropper will have a
vanishing probability of learning their key after $O(N^{1.5-\eps})$ queries to the same oracle.
If~we take account of computational complexity in addition to query complexity, we must
be content with \mbox{$k=2$}, in which case the claim is much more modest, but still the
quantum codebreaker must work more than linearly harder than the classical codemakers.
Along the way, we need to develop in Section~\ref{sec:avgadv} new tools for the study of
\emph{average-case} quantum query complexity, which had essentially remained virgin
territory despite its obvious importance, in~particular but not only for cryptography.

The second modifications to Merkle's original scheme that has been
considered~\cite{ICQNM,bhkkls:crypto,bhkkls:arxivv2} is to
play a fair game in allowing the codemakers to use quantum computers as well.
The~first benefit is that we can enlarge the domain of $f$ to contain $N^3$ points. 
If~Alice proceeds exactly as before,
Bob can use an extension of Grover's algorithm known as BBHT~\cite{BBHT}
in order to find random preimages of the $N$ image points initially sent by Alice
at the cost of \mbox{$O(\sqrt{N^3/N}\,) = O(N)$} queries per preimage, provided
\mbox{$k \ll N$}.
This increase in the domain size of $f$, and correspondingly of~$t$,
makes it significantly harder for a quantum eavesdropper
to solve the conundrum and discover the key shared by Alice and Bob.
Indeed, we also prove Theorem~\ref{main-quant-thm}, stated below,
to the effect that no cryptanalytic attack can succeed on a typical instance of
the protocol, except with vanishing probability, short of making
$\Omega(N^{1+k/(k+1)})$ queries to the black-box functions.
Again, this theorem was claimed in Ref.~\cite{bhkkls:arxivv2} but its proof
was fundamentally flawed for~\mbox{$k>2$}.
Taking $k$ sufficiently large, this offers a quantum-against-quantum security
that is arbitrarily close to the quadratic security that the original scheme of
Merkle~\cite{merkle74} offered in the classical-against-classical scenario.
The~second benefit to allowing the codemakers to use quantum computers
is that now a quantum Alice 
can be efficient in terms of
computation time, in addition to query complexity, even when \mbox{$k=3$}.
According to Theorem~\ref{main-quant-thm}, we get an $\Omega(N^{7/4})$
security guarantee for a protocol that could become practical once sufficiently powerful
quantum computers start to seriously threaten
the security of the current Internet cryptographic infrastructure.
This is the most secure \emph{proven} solution ever discovered to the conundrum of post-quantum
cryptography~\cite{Harachov} when all parties have equal quantum computing capabilities,
at least in the random oracle model, and its security
is reasonably close to that of Merkle's provably optimal scheme in an all-classical world
but otherwise in the same model.

\subsection{The \boldmath{\ksum} Problem}

The security of the protocols that we study is based on the $\ksum$ problem, which
consists in searching for $k$ elements among $N$ in some abelian group $\Gr$
whose sum is a given value~\mbox{$w \in \Gr$}\@.

\begin{definition}[\boldmath{\ksum} problem]\label{def:ksum}
Given an abelian group $\Gr$, a function \mbox{$t:D \rightarrow \Gr$} for some domain~$D$,
a \emph{target} \mbox{$w \in \Gr$} and $N$ distinct elements \mbox{$x_1, x_2, \ldots, x_N \in D$},
the problem is to find $k$ indices \mbox{$1 \le i_1 < i_2 < \cdots < i_k \le N$} such that
\mbox{$w=\sum_{j \in 1}^k t(x_{i_j})$},
provided a solution exists.
The \emph{decision} version of \ksum{} is to decide whether or not a solution exists.
\end{definition}
It is {crucial} to understand that we are not interested in how much computation \emph{time}
would be required to find a solution, if one exists. Rather, we want to minimize the
\emph{number of calls} to function $t$ that will be required.
Naturally, a quantum algorithm is allowed to query~$t$ on superpositions of elements of~$D$.

When~\mbox{$k=1$}, this is simply the \emph{unstructured search problem},
which consists in finding $i$ such that \mbox{$t(x_i)=w$}, provided it exists.
When~\mbox{$k=2$} and $\Gr$ is the group of bit strings of a given length
under bitwise exclusive-or, $\ksum$ takes the name of~$\twoxor$.
In~turn, when \mbox{$w=0$}, $\twoxor$ becomes the search version of the
Element Distinctness (\ED) problem, which consists in finding a collision
in a given function if it is not one-to-one.

\begin{definition}[Element Distinctness (\ED) problem]\label{def:ED}
Given a function \mbox{$t:D \rightarrow R$}, the \emph{decision} element distinctness (\ED) problem
is to decide whether or not this function is one-to-one.
\end{definition}

\begin{definition}[Search version of \ED]
Given a function \mbox{$t:D \rightarrow R$}, the \emph{search} version of
the element distinctness problem (\SED) is to find a pair of distinct \mbox{$x,x' \in D$}
such that \mbox{$t(x)=t(x')$}, provided such a pair exists.
\end{definition}

Quantum lower bounds have been proved on all these
problems~\cite[etc.]{AaronsonShi,spalek:kSumLower}, but only in the worst-case
scenario, which is most frequently studied in the field of computational and query complexity.
For~some of these problems, such as \ED, \SED, \twoxor{} and \twosum,
a simple \emph{classical} randomized reduction suffices for proving their difficulty
on average from their difficulty in the worst case even in the quantum
setting, at least if we add the promise that if there is a solution, then it is
unique.
However, this does not appear
to be the case for \ksum{} when~\mbox{$k>2$}. Our~main mistake in Ref.~\cite{bhkkls:arxivv2}
was to take such a reduction for granted for arbitrary $k$ after having nearly proved it
in the case~\mbox{$k=2$}. ``Nearly'' because the proof for \mbox{$k=2$} was flawed,
albeit easy to repair. Not~so for \mbox{$k>2$}, however.
In~order to prove the security of the key establishment protocols described above
in a cryptographically meaningful context, we need to prove the difficulty of
\ksum{} on average for arbitrary~$k$, which requires new quantum lower bound
techniques. In~fact, we need to prove the difficulty on average of a \emph{composed}
version of~\ksum{}, defined below in Section~\ref{sec:cksum}, which does not follow by a classical reasoning
from the average difficulty of plain~\ksum{}. Therefore, we also have to develop a new
composition theorem that works on average as well.

The first quantum lower bound discovered among these problems was for the decision
element distinctness problem. Aaronson and Shi~\cite{AaronsonShi} proved that this problem requires
$\Omega(d^{2/3})$ queries to~$t$ in the worst case, where $d$ is the cardinality
of domain~$D$.
There was a technical condition in their original proof that required
\blue{\mbox{$r \ge d^2$}, where $r$ is}
the cardinality of range~$R$,
but that condition was subsequently lifted~\blue{\cite{amb03,Kutin}}.
Later, Belovs and \v{S}palek~\cite{spalek:kSumLower} proved that solving \ksum{} requires
$\Omega(N^{k/(k+1)})$ queries to~$t$ in the worst case, provided \mbox{$m \ge N^k$},
where $m$ is the order of group~$\Gr$
\blue{and $N$ is as in Definition~\ref{def:ksum}}.

Even though the technique used by Aaronson and Shi was adequate only to prove worst-case
lower bounds, it is elementary to conclude by a classical reasoning that the hardness in worst-case
of \ED{} implies the same hardness on average for \ED, \SED{} and \twoxor.
But, as we said already, a completely new technique, which we develop in Section~\ref{sec:avgadv},
is required to prove a matching hardness result for \ksum{} on average, which is stated as
Theorem~\ref{thm:1average} in Section~\ref{sec:ksum}.

However, even this is not sufficient to derive the security of the key establishment protocols
described above in a cryptographically meaningful manner.
Indeed, the eavesdropper is not faced with an instance of \ksum{}, as specified in
Definition~\ref{def:ksum}.
He~learns the value of $w$ when Bob transmits it to Alice,
and he has access to black-box function~$t$, but he does not know the $x$'s,
which are kept secret by Alice. Instead, he learns the image of those $x$'s by function~$f$,
which we called the $y$'s, when Alice sent them to Bob in the first step of the protocol.
In~fact, he has to solve the more difficult \emph{Hidden}  \ksum{} problem,
which we now proceed to describe.

\subsection{Hidden and Composed \boldmath{\ksum} Problems}\label{sec:cksum}
\label{sec:hidden}

The hidden \ksum{} problem, defined below, corresponds precisely to the task
facing the eavesdropper.

\begin{definition}[Hidden \boldmath{\ksum} problem]\label{def:hksum}
Given two sets $D$ and $R$, an abelian group $\Gr$,
two functions \mbox{$f:D \rightarrow R$} and \mbox{$t:D \rightarrow \Gr$},
$N$ distinct elements \mbox{$y_1, y_2, \ldots, y_N \in \Image(f)$},
and a \emph{target} \mbox{$w \in \Gr$}, 
the problem is to find $k$ indices \mbox{$1 \le i_1 < i_2 < \cdots < i_k \le N$}
and a preimage \smash{$x_{i_j}$} under $f$ for each $y_{i_j}$, \mbox{$1 \le j \le k$},
meaning that \mbox{$f(x_{i_j})=y_{i_j}$},
such that
\mbox{$w=\sum_{j=1}^k t(x_{i_j})$},
provided a solution exists.
The \emph{decision} version of hidden \ksum{} is to decide if a solution exists.
\end{definition}

In order to prove lower bounds on the quantum cryptanalytic task of breaking typical
\mbox{instances} of
the protocols described in Section~\ref{sec:variations}, we proceed in two steps.
First we have to prove the hardness of the hidden \ksum{} problem on average.
Then, we have to exhibit a reduction that shows how to solve an average instance of
the hidden \ksum{} problem using an adversary who thinks he is breaking
a typical instance of the key establishment protocol.
To prove the hardness of the hidden \ksum{} problem on average,
it helps to consider a more structured version of it, which is given by
the composition of \ksum{} with a search problem called \psearch{}, defined below.

\begin{definition}[\psearch{} problem] \label{def:psearch}
Let $A$ be some set and $\star$ a symbol not in~$A$.
Consider the set $P$ of strings $(a_1, \ldots, a_{\innern})$ in $(A\cup\{\star\})^{\innern}$
with the promise that exactly one value is not~$\star$.
The problem $\psearch_\innern: P \rightarrow A$ consists in finding this non-$\star$
value by making queries that take
$i$ as input and return $a_i$, $1\leq i \leq \innern$.
\end{definition}

An \blue{equivalent} formulation of the \ksum{} problem would consist in
a target $w$ in abelian group $\Gr$ and a list
\mbox{$(t_1, t_2, \ldots, t_N)$} of elements of~$\Gr$.
The problem is to find $k$ indices \mbox{$1 \le i_1 < i_2 < \cdots < i_k \le N$} such that
\mbox{$w=t_{i_1}+t_{i_2}+\cdots+t_{i_k}$}. We~are charged for accessing each~$t_i$
given~$i$. This is equivalent to Definition~\ref{def:ksum} simply by taking
\mbox{$t_i=t(x_i)$}, but it is more convenient since it allows us to consider the
composition of \ksum{} with $N$ instances of \psearch.
Thus we define the \emph{Composed} version of \ksum{} as follows.

\begin{definition}[Composed \boldmath{\ksum} problem]\label{def:cksum}
Given a target $w$ in abelian group $\Gr$ and  $N$ \mbox{instances} of the $\psearch_\innern$
problem using $\Gr$ as set~$A$, we want to solve the \ksum{} problem with $t_i$
being the only non-$\star$ element in the $i^{\mathrm{th}}$ instance of
$\psearch_\innern$. Said otherwise, this is the composition of \ksum\ and $\psearch_\innern$ denoted
\mbox{$\ksum \comp\psearch_\innern^N$}.
\end{definition}

The composed \ksum{} problem (Definition~\ref{def:cksum}) is similar to its 
hidden variant (Definition~\ref{def:hksum}), except that it is more structured,
hence easier.
Specifically, the $x_i$'s that serve to define \mbox{$t_i = t(x_i)$} in the hidden version,
\mbox{$1 \le i \le N$},
can be \emph{a~priori} any element of~$D$,
whereas they are put in $N$ ``buckets'' of size $\innern$ in the composed version.
If~we choose the size of $D$ to be the product of $N$ and~$\innern$,
any algorithm capable of solving the hidden version can serve directly
to solve the composed version simply by taking no account of the
additional information provided by the buckets. Moreover, a random instance
of the composed version can be transformed into a random instance of
the hidden version, essentially by mixing the buckets. It~follows that any
lower bound on the composed problem translates directly into the same
lower bound on the hidden problem, \emph{mutatis mutandis}.

In~Sections~\ref{sec:avgadv} to~\ref{sec:composition}, which are more technical,
we give a lower bound on the composed problem in a series of steps.
First, we give a new general method to prove lower bounds for the average-case quantum query complexity (Section~\ref{sec:avgadv}). 
This method is closely related to the technique given in Ref.~\cite{belovs:variations},
albeit with essential differences. 
Second, building on techniques from Refs~\cite{spalek:kSumLower,belovs:onThePower}, 
we show a lower bound on the average-case quantum query complexity of \ksum~(Section~\ref{sec:ksum}).
Third, we show a composition theorem for average-case quantum query complexity, which allows 
us to conclude with Theorem~\ref{thm:bucket} (Section~\ref{sec:composition}).

When we apply this theorem with the parameters that correspond to the
protocols described in Section~\ref{sec:variations}, we should take
\mbox{$\outern=N$}, which is the number of images sent by Alice in the first step
of any of these protocols and therefore also the number of buckets.
Furthermore, we should take the product of $\innern$,
the size of the buckets, with $\outern$, the number of buckets, to
correspond to the size of the domain $D$ used in the protocols.

Putting it all together, Theorem~\ref{thm:bucket} gives us the following
lower bound on the difficulty to solve the hidden \ksum{} problem if the
domain $D$ of functions $f$ and $t$ contains $d$ elements.

\begin{theorem}\label{th:hiddenksum}
Any quantum algorithm that uses at most $T$ queries to
find a solution to the hidden \ksum{} problem with success probability
at least \mbox{$\epsN_N>0$} on average over the uniform distribution on
positive instances requires
\[
\frac T{\epsN_N} = \Omega\sB[\sqrt{d/N-1}\; N^{k/(k+1)} ]
\]
provided $m= \omega\Big(N^{k+\tfrac 2 {k+1}} \Big)$,
where $m$ is the order of the underlying abelian group.
\end{theorem}

\subsection{The Security of Key Establishment}
\label{sec:reduction}

We proved (correctly!)\ in Ref.~\cite{bhkkls:arxivv2} that any eavesdropper who
succeeds in obtaining the key with non-vanishing success probability $\epsN$ in any of the protocols described in Section~\ref{sec:variations},
after making no more than $T$ queries, on average over the runs 
of the protocol, can be used to solve the hidden \ksum{} problem with
the same parameters. Therefore, using the fact that \mbox{$d=N^2$} for the
classical protocols and \mbox{$d=N^3$} for the quantum protocol, 
we can apply Theorem~\ref{th:hiddenksum} to conclude that the protocols are secure
according to the following theorems.

\begin{theorem}\label{main-class-thm}
Any quantum eavesdropping strategy that makes
\mbox{$o\big(N^{\frac12+\frac{k}{k+1}}\big)$}
queries to the black-box functions against a typical run of the classical protocol using parameter~$k$
will fail to recover the key, except with vanishing probability.
\end{theorem}

\begin{theorem}\label{main-quant-thm}
Any quantum eavesdropping strategy that makes
\mbox{$o\big(N^{1+\frac{k}{k+1}}\big)$}
queries to the black-box functions against a typical run of the quantum protocol using parameter~$k$
will fail to recover the key, except with vanishing probability.
\end{theorem}

\noindent
Furthermore, we showed in Ref.~\cite{bhkkls:arxivv2} that these bounds are tight.

\section{Average-Case Quantum Adversary Lower Bound Method}
\label{sec:avgadv} 

We generalize the adversary lower bound method to handle average-case complexity.
A similar bound from \blue{Ref.}~\cite{belovs:variations} already gives a lower bound technique on 
average-case query complexity,
but it cannot be applied directly here, as we explain below.

We use the following complexity measure, closely related to the adversary bound~\cite{amb02,hoyer:advNegative}.
We~give a formulation tailored to the following problem. Given two distributions $\cP$ and $\cQ$, and
an algorithm that attempts to distinguish between them, we consider the number of queries
this algorithm must make in order to succeed. The algorithm is given one input, and
accepts if it thinks the 
sample it is given comes from $\cP$ and rejects otherwise. 
The measure of success is given by the probabilities $s_\cP$ and $s_\cQ$,
which are the probability of accepting when the algorithm is given samples from $\cP$ and $\cQ$, respectively.  

\begin{definition}
\label{def:adv}
Let $\cP$ and $\cQ$ be two probability distributions on $\cD$, 
and $p_x$ and $q_y$ denote probabilities of $x$ and $y$ in $\cP$ and $\cQ$, respectively.
Let $s_\cP, s_\cQ$ be real numbers in $[0,1]$ (representing the acceptance probability on 
distributions $\cP$ and $\cQ$, respectively).
For a given matrix $\Gamma$, define the adversary bound with respect to $\Gamma, \cP, s_\cP,\cQ, s_\cQ$ as
\begin{equation}
\label{eqn:adv}
\Adv(\Gamma;\cP,s_\cP;\cQ,s_\cQ) = \Omega\s[ \min_{j\in[n]} \frac{\delta_\cP^*\Gamma\delta_\cQ\* -  \tau(s_\cP,s_\cQ) \|\Gamma\|} {\|\Gamma\hadamard \Delta_j \|} ].
\end{equation}
Here, $\hadamard $ denotes entrywise (or~Hadamard) product, and
$\|A\|$ denotes the spectral norm of $A$ (which is equal to its largest singular value).
The vectors $\delta_\cP[x] = \sqrt{p_x}$ and $\delta_\cQ[y]=\sqrt{q_y}$ are unit vectors in $\bR^\cD$;
for $j\in[n]$, the $|\cD|\times |\cD|$ matrix  $\Delta_j$ is defined by $\Delta_j[x,y] = 1_{x_j\ne y_j}$; and
\begin{equation}
\label{eqn:dpq}
\tau(s_\cP,s_\cQ) = \sqrt{\strut s_\cP s_\cQ} + \sqrt{\strut (1-s_\cP)(1-s_\cQ)}.
\end{equation}
\end{definition}

\begin{theorem}
\label{thm:adv}
Assume $\cA$ is a quantum algorithm that makes $T$ queries to the 
input string $x=(x_1,\dots,x_n)\in\cD$, and then either accepts or rejects.
Let $\cP$ and $\cQ$ be two probability distributions on $\cD$. 
Let $s_\cP$ and $s_\cQ$ be acceptance probability of $\cA$ when $x$ is sampled from $\cP$ and $\cQ$, respectively.  Then,
\begin{equation*}
T \geq \Adv(\Gamma; \cP, s_\cP; \cQ, s_\cQ),
\end{equation*}
for any $|\cD|\times|\cD|$ matrix $\Gamma$.
\end{theorem}

If $\cP$ and $\cQ$ have partial supports, then we may use a matrix $\Gamma$ whose rows
are indexed by elements in the support of $\cP$ and columns by elements of the support of $\cQ$.
In that case we can extend the matrix $\Gamma$ by adding all-0 rows and columns. Notice that this does not alter the value of $\Adv$.

First let us consider why we need two distributions $\cP, \cQ$ on the inputs 
(and why we cannot use existing techniques such as
\blue{Theorem~33 from Ref.~\cite{belovs:variations}} for
decision problems, where $\cP=\cQ$). The distribution we care about is the uniform distribution over the 
positive instances. Under this distribution, the decision problem is of course trivial.
Using this distribution as both $\cP$ and $\cQ$ as in \blue{Ref.}~\cite{belovs:variations} would give a trivial bound.

Instead, Theorem~\ref{thm:adv} gives a lower bound on the query complexity of
an algorithm that \mbox{attempts} to distinguish between two distributions~$\cP$ and $\cQ$.
Taking $\cP$ as the uniform distribution over positive instances, and $\cQ$ as the uniform 
distribution
over all instances implies a lower bound for the {\em search} problem of 
finding $k$ elements that sum to $w$ with the promise that the instance is positive,
by the following argument. Assume an algorithm solves the search problem with
$T$ queries with non-vanishing probability. Then we can transform this algorithm into 
a distinguishing algorithm with one-sided error: if the algorithm outputs a candidate solution
$a_1,\ldots,a_k$, make $k$ additional queries and check that they sum to $w$.
If they do, accept, else reject. Then the acceptance probability on 
negative instances is 0. Since most instances are negative, the acceptance 
probability on the
uniform distribution is close to~0. We~\blue{are} interested in the acceptance 
probability on the positive instances, as a function of the number of queries~$T$.

We now proceed to the proof of Theorem~\ref{thm:adv}.
Our proof is closely related the proof of the worst-case negative-weighted adversary bound
from Ref.~\cite{hoyer:advNegative}.  We follow a slightly simplified version of the proof
from Ref.~\cite{belovs:phd}.
As usual, we introduce a progress function, show that initially, the progress function is large (Claim~\ref{clm:adv1}), at the end, it is small (Claim~\ref{clm:adv2}), and that at each step, the 
decrease is bounded (Claim~\ref{clm:adv3}).

\begin{proof}[Proof of Theorem~\ref{thm:adv}]
Recall that a quantum query algorithm is given by the following sequence of operations
\[
U_0\to O_x\to U_1\to O_x\to U_2\to  \dots \to U_{T-1} \to O_x\to U_T,
\]
where $O_x$ denotes  the input oracle, and the $U_i$s are arbitrary 
unitary transformations.
The operator $O_x$ is defined by $O_x \ket|a>\ket|i>= \ket |a+x_i>\ket|i>$
which can be decomposed as
\begin{equation}
\label{eqn:Oz}
O_x = \bigoplus_{j=0}^n O_{x_j},
\end{equation}
where for $b\in\Gr_m$, 
$O_b\colon \ket|a>\ket|i> \mapsto \ket |a+b>\ket|i>$. 
The addition in the first register is the group operation of $\Gr_m$.

For an integer $t$ between 0 and $T$, and $x\in\cD$, let
\begin{equation}
\label{eqn:advpsit}
\psi^{(t)}_x = U_tO_x U_{t-1} O_x\cdots U_1 O_x U_0 \ket |0>.
\end{equation}
be the state of the algorithm on the input $x$ after $t$ queries.
We define the quantity called the \emph{progress function} as follows
\begin{equation}
\label{eqn:progress}
W^{(t)} = \sum_{x,y\in\cD} \sqrt{p_x q_y}\;\Gamma[x,y] \ipA<\psi^{(t)}_x, \psi^{(t)}_y>.
\end{equation}

The proof is split into three parts: proving that $W^{(0)}$ is large, and that both $W^{(T)}$ and  $W^{(t)}-W^{(t+1)}$ are small.

\samepage{
\begin{clm}
\label{clm:adv1}
$W^{(0)}=\delta_\cP^* \Gamma \delta_\cQ\*$.
\end{clm}
\begin{proof}
We have $\psi^{(0)}_x = U_0\ket |0>$ no matter what $x$ is.  Hence, $\ipA<\psi^{(0)}_x, \psi^{(0)}_y>=1$ for all $x, y\in\cD$.  Plugging this into \blue{Eq.~\ref{eqn:progress}} gives
\[
W^{(0)} = \sum_{x,y\in\cD} \sqrt{p_xq_y}\; \Gamma[x,y]= \delta_\cP^* \Gamma \delta_\cQ\*. \] 
\par\mbox{}\\[-8ex]
\end{proof}
}

\vspace{1ex}

Before we proceed, we need a simple result from linear algebra.
\begin{lemma}
\label{lem:simple}
Let $A$ be $n\times n$ matrix, and $U$ and $V$ be $m\times n$ matrices with columns $\{u_i\}_{i\in[n]}$ and $\{v_i\}_{i\in[n]}$, respectively.  Then,
\[
\absB|\sum_{i,j\in[n]} A[i,j]\ip<u_i, v_j>| \le \norm|A| \normFrob|U| \normFrob|V|.
\]
\end{lemma}

\begin{proof}
Using the Cauchy-Schwarz inequality, and the definition of the spectral norm:
\[
\absB|\sum_{i,j\in[n]} A[i,j]\ip<u_i,v_j>| =
|\ip<A,U^*V>| = |\ip<U A,V>| \le \normFrob|U A|\normFrob|V| \le \norm|A| \normFrob|U| \normFrob|V|.
\]
\par\mbox{}\\[-8ex]
\end{proof}

\vspace{1ex}

\begin{clm}
\label{clm:adv2}
$W^{(T)}\le \sB[\sqrt{\strut s_\cP s_\cQ} + \sqrt{\strut (1-s_\cP)(1-s_\cQ)}\;] \norm|\Gamma|$.
\end{clm}
\begin{proof}
Denote for brevity $\psi_x = \psi^{(T)}_x$.
Also, let $\sfigA{\Pi_0, \Pi_1}$ be the final measurement of the query algorithm $\cA$.
We have
\begin{equation}
\label{eqn:Wt}
W^{(T)} = 
\sum_{x,y\in\cD} \sqrt{p_xq_y}\;\Gamma[x,y] \ipA<\Pi_0\psi_x, \Pi_0\psi_y>
+ \sum_{x,y\in\cD} \sqrt{p_xq_y}\;\Gamma[x,y] \ipA<\Pi_1\psi_x, \Pi_1\psi_y>.
\end{equation}

Let us estimate the first term of Eq.~\ref{eqn:Wt}.
Denote by $U$ and $V$ the matrices having $u_x = \sqrt{p_x} \psi_x$ and $v_y = \sqrt{q_y} \psi_y$ as their columns, respectively.  Then, by Lemma~\ref{lem:simple}, the first term of
\blue{Eq.~\ref{eqn:Wt}} is at most $\|\Gamma\| \normFrob|U|\normFrob|V|$, where
\[
\normFrob|U|^2 = \sum_{x\in\cD}p_x \|\Pi_0\psi_x\|^2 = 1-s_\cP,
\qquad\text{and}\qquad
\normFrob|V|^2 = \sum_{y\in\cD}q_y \|\Pi_0\psi_y\|^2 = 1-s_\cQ.
\]
Hence, the first term of Eq.~\ref{eqn:Wt} is at most $\sqrt{(1-s_\cP)(1-s_\cQ)}\norm|\Gamma|$.
Similarly, the second term is at most $\sqrt{s_\cP s_\cQ} \norm|\Gamma|$.
Adding them up, we get the required inequality. 
\end{proof}

\vspace{1ex}

\begin{clm}
\label{clm:adv3}
$|W^{(t)}-W^{(t+1)}|\le 2\max_{j\in[n]} \norm|\Gamma\hadamard \Delta_j|$.
\end{clm}
\begin{proof}
Denote $\psi_x = \psi^{(t)}_x$ and $\psi'_x = \psi^{(t+1)}_x$.  
The vector $\psi_x$ can be decomposed as $\bigoplus_{j= 0}^{n} \psi_{x,j}$ where the decomposition is the same as for $O_x$ in Eq.~\ref{eqn:Oz}.
If $x_j=y_j$, then the input oracle does not change the inner product between $\psi_{x,j}$ and $\psi_{y,j}$, hence, the corresponding entry of $\Gamma$ can be ignored.  Formally, for any $x,y\in\cD$, we have

\[
\ip<\psi_x, \psi_y>-\ip<\psi'_x, \psi'_y>  = 
\ip<\psi_x, \psi_y>-\ip<O_x\psi_x, O_y\psi_y> = 
\sum_{j=0}^n \chi_{x,y,j},
\]
where $\chi_{x,y,j} = \ip<\psi_{x,j} , \psi_{y,j}> - \ip<O_{x_j}\psi_{x,j}, O_{y_j}\psi_{y,j} >$.  Note that $\chi_{x,y,j} = 0$ if $x_j=y_j$.  In particular, $\chi_{x,y,0}=0$.  
Thus,
\begin{eqnarray}
\lefteqn{|W^{(t)}-W^{(t+1)}| }\notag \\
	& = & \absC|\sum_{x,y\in\cD} \sqrt{p_xq_y}\;\Gamma[x,y] \sB[\ip<\psi_x, \psi_y> - \ip<\psi'_x, \psi'_y>]| \notag\\
	& = &\absC| \sum_{x,y\in\cD}\sum_{j=0}^n\sqrt{p_xq_y}\; \Gamma[x,y] \chi_{x,y,j} |\notag \\
	& = & \absC|\sum_{j=1}^n \sum_{x,y\in\cD}\sqrt{p_xq_y}(\Gamma\hadamard \Delta_j) [x,y]\chi_{x,y,j} |\notag\\
	& \le & \sum_{j=1}^n \absC| \sum_{x,y\in\cD} \sqrt{p_xq_y}(\Gamma\hadamard \Delta_j)[x,y]\ip<\psi_{x,j} , \psi_{y,j}>  | \notag \\
	&  &\qquad +
\sum_{j=1}^n \absC| \sum_{x,y\in\cD} \sqrt{p_xq_y}(\Gamma\hadamard \Delta_j)[x,y] \ip<O_{x_j}\psi_{x,j} , O_{y_j}\psi_{y,j}> |. \label{eqn:Wtchange}
\end{eqnarray}
Let us estimate the second term, the first one being similar.  
For $j\in[n]$, let $U_j$ be the matrix with columns $u_{j,x} = \sqrt{p_x} O_{x_j}\psi_{x,j}$,
and $V_j$ be the matrix with columns $v_{j,y} = \sqrt{q_y} O_{y_j}\psi_{y,j}$.
By~Lemma~\ref{lem:simple} and the Cauchy-Schwarz inequality, the second term of
\blue{Eq.~\ref{eqn:Wtchange}} is at most
\begin{eqnarray*}
\sum_{j=1}^n \norm|\Gamma\hadamard \Delta_j| \normFrob|U_j|\normFrob|V_j| 
&\le &\max_{j\in [n]} \norm|\Gamma\hadamard \Delta_j| \sum_{j=1}^n  \normFrob|U_j|\normFrob|V_j|\\
&\le &\max_{j\in [n]} \norm|\Gamma\hadamard \Delta_j| \sqrt{\sC[\sum_{j=1}^n \normFrob|U_j|^2 ]\sC[\sum_{j=1}^n \normFrob|V_j|^2]}.
\end{eqnarray*}
Also, we have
\[
\sum_{j=1}^n  \normFrob|U_j|^2 = \sum_{j=1}^n \sum_{x\in\cD} p_x \norm|O_{x_j}\psi_{x,j}|^2 =
\sum_{x\in\cD} p_x \sum_{j=1}^n  \norm|\psi_{x,j}|^2 \le
\sum_{x\in\cD} p_x \norm|\psi_{x}|^2 = \sum_{x\in\cD} p_x = 1,
\]
and, similarly, $\sum_{j=1}^n  \normFrob|V_j|^2\le 1$.  Combining the last three inequalities, we get that the second term of Eq.~\ref{eqn:Wtchange} is at most $\max_{j\in [n]} \norm|\Gamma\hadamard \Delta_j|$.
Using the same estimate for the first term, we obtain the required inequality. 
\end{proof}

\vspace{1ex}

This concludes the proof of Theorem~\ref{thm:adv}
\end{proof}

\section{Average-Case Complexity of \boldmath{\ksum}}
\label{sec:ksum}

Recall the 
$\ksum$  problem on $n$ elements in an abelian group 
$\Gr_m$ where $m$ is the order of the group.
Let $w$ be a fixed element of $\Gr_m$.
An input $x=(x_1,\dots,x_n)$ is called \emph{positive} if there exists a $k$-subset $\blue{V} = \{t_1,\dots,t_k\}\subseteq [n]$ such that $x_{t_1}+\cdots+x_{t_k} = w$ in $\Gr_m$.
Otherwise, the input is called \emph{negative}.

Consider the following probability distribution $\cP$ on positive inputs:
\itemstart
\item Select a $k$-subset $\blue{U}$ of $[n]$ uniformly at random;
\item assign to $\blue{U}$ a uniformly random string in $\Gr_m^{|\blue{U}|}$ whose sum is $w$;
\item choose the remaining elements uniformly at random.
\itemend

\begin{theorem}
\label{thm:1average}
Assume $\cS$ is a quantum algorithm for the search problem $\ksum$ 
that makes $T$ queries and succeeds  with probability~$\epsn>0$ over  inputs sampled from the distribution $\cP$.  
Then, 
\[
\frac T\epsn = \Omega\s[n^{k/(k+1)} ],
\] 
provided that $\epsn = \omega\s[n^{-1/(k+1)}]$ and $m = \Omega\s[n^{k+\frac2{k+1}} ]$
is again the order of the underlying abelian group.

\end{theorem}

This theorem uses the following claim.
\begin{clm}
\label{clm:2average}
Let the distribution $\cP$ be as above, and $\cQ$ be the uniform distribution on all the inputs.
There exists a matrix $\Gamma$ satisfying the following constraints:
\[
\delta^*_\cP \Gamma\delta_\cQ\* = n^{k/(k+1)},\qquad
\norm|\Gamma| \le \sB[{1+O\sA[n^{-1/(k+1)}] }] n^{k/(k+1)},
\quad\text{and}\qquad
\norm|\Gamma\hadamard \Delta_j| = O(1)
\]
in the notation of Theorem~\ref{thm:adv}.
\end{clm}
\begin{proof}
Our construction is the same as in Refs~\cite{spalek:kSumLower,belovs:onThePower}, 
which we sketch below.
The matrix $\Gamma$ consists of $\binom n k$ matrices $G_\blue{V}$ stacked one on another for all possible choices of $\blue{V}=\{t_1,\dots,t_k\}\subset[n]$:
\looseness=-1
\begin{equation}
\label{eqn:Gamma}
\Gamma = \left(
\begin{array}{c} G_{1,2,\dots,k} \\ G_{1,2,\dots,k-1,k+1} \\ \dots \\ G_{n-k+1,n-k+2,\dots,n} \\ \end{array}
\right).
\end{equation}
For $\blue{V} = \{t_1,\dots,t_k\}$, $G_\blue{V}$ is \blue{a} $\cP_\blue{V} \times \cD$ matrix, where $\cD = \Gr_m^n$ and $\cP_\blue{V} = \{x\in\cD \mid x_{t_1}+\cdots+x_{t_k} = w\}$.
Note that the uniform distribution on the labels of the rows and the columns of $\Gamma$ \blue{generates} the probability distributions $\cP$ and $\cQ$, respectively.

Let $E_0 = J_m/m$ be the $m\times m$ matrix with all entries $1/m$, and let $E_1 = I - E_0$.
For $\blue{U}\subseteq [n]$, define a $\cD\times\cD$ matrix
$E_\blue{U} = \bigotimes_{j\in[n]} E_{s_j}$, where $s_j=1$ if $j\in \blue{U}$ and $s_j=0$ otherwise.
The matrices $G_\blue{V}$ in Eq.~\ref{eqn:Gamma} are given by
\begin{equation}
\label{eqn:GT}
G_\blue{V} = \sqrt m \sum_{\blue{U}\subseteq [n]: \blue{V}\not\subseteq \blue{U}} \alpha_{|\blue{U}|}\; E_\blue{U}^{\cP_\blue{V}, \cD},
\end{equation}
where
\begin{equation}
\label{eqn:alpha}
\alpha_\ell = {n\choose k}^{-1/2} \max\sfigB{ n^{k/(k+1)} - \ell,\; 0 }.
\end{equation}
This finishes the definition of the matrix $\Gamma$.
From \blue{Refs}~\cite{spalek:kSumLower,belovs:onThePower}, we have the following estimate:
\begin{clm}
\label{clm:GammacircDelta}
For $\Gamma$ defined above and $j\in[n]$, we have $\norm|\Gamma\hadamard \Delta_j| = O(1)$.
\end{clm}
Thus, it remains to prove the first two statements of Claim~\ref{clm:2average}.  For that, we need a slightly more careful analysis than in \blue{Ref.}~\cite{spalek:kSumLower}.

\mycommand{supp}{\mathop{\mathrm{supp}}\nolimits}
Let $e_0,\dots,e_{m-1}$ be the Fourier basis of $\bC^{\Gr_m}$.
Recall that it is an orthonormal basis given by $e_i[j] = \frac1{\sqrt m}\omega^{ij}$, where $\omega = \ee^{2\pi\ii/m}$.
For $v=(v_1,\dots,v_n)\in \cD$, define  vector
$e_v = e_{v_1}\otimes e_{v_2}\otimes\cdots\otimes e_{v_n}$.
These vectors form the Fourier basis of \smash{$\bC^{\cD}$}.
The \emph{support} of $v$ is defined as $\supp(v) = \{ i\in[n]\mid v_i\ne 0 \}$.
The \emph{weight} of $v$ is defined as the size of the support: $|v| = |\supp(v)|$.

Let $\blue{V}=\{t_1,\dots,t_k\}\subseteq[n]$ be fixed, and denote $t=t_1$ for brevity.
We can identify the sequences in $\cP_\blue{V}$ with the sequences in $\cD' = \Gr_m^{[n]\setminus\{t\}}$, since any sequence $x$ in $\cD'$ can be uniquely extended to a sequence in $\cP_\blue{V}$ using condition $x_{t_1}+\cdots+x_{t_k} = w$.
Note that for any $v\in\cD$ and $x\in\cP_\blue{V}$, we have
\[
e_v[x] =\frac 1 {m^{n/2}} \omega^{ \sum_{j\in [n]} v_jx_j } = \frac 1 {m^{n/2}} \omega^{wv_t} \omega^{ \sum_{j\in \blue{V}} (v_j-v_t)x_j + \sum_{j\notin \blue{V}} v_jx_j }.
\]
Hence,
\begin{equation}
\label{eqn:ESelem}
\sqrt m E_\blue{U}^{\cP_\blue{V}, \cD} e_v =
\begin{cases}
\omega^{wv_{t}}e_{v'} & \text{if $\supp(v) = \blue{U}$;}\\
0 & \text{otherwise;}
\end{cases}
\end{equation}
where $e_{v'}$ is the element of the Fourier basis of $\bC^{\cD'}$ defined by
\[
v'_j = 
\begin{cases}
v_j - v_t &\text{for $j\in \blue{V}\setminus\{t\}$;}\\
v_j &\text{for $j\in [n]\setminus \blue{V}$.}
\end{cases}
\]
This suggests the following definition: we say that $v\in\Gr_m^n$ can be obtained from $u\in\Gr_m^n$ by a \emph{shift} in $\blue{V}$ if there exists $a\in\Gr_m$ such that
\[
v_j = 
\begin{cases}
u_j +a &\text{for $j\in \blue{V}$;}\\
u_j &\text{for $j\in [n]\setminus \blue{V}$.}
\end{cases}
\]
Equations~\ref{eqn:ESelem} and~\ref{eqn:GT} imply that for sequences $u,v\in\cD$, we have
\begin{equation}
\label{eqn:GTeu}
\absA|\ip<G_\blue{V} e_u,  G_\blue{V} e_v>| =
\begin{cases}
\alpha_{|u|}\alpha_{|v|} &\parbox{.5\textwidth}{if $\blue{V}\not\subseteq\supp(u)$, $\blue{V}\not\subseteq\supp(v)$ and $v$ can be\\ obtained from $u$ by a shift in $\blue{V}$;}\\
0&\text{otherwise.}
\end{cases}
\end{equation}

Now we are ready to continue with the proof of all statements in Claim~\ref{clm:2average}.

\begin{clm}
\label{clm:deltaGammadelta}
$\delta_\cP^* \Gamma\delta_\cQ\* = n^{k/(k+1)}$.
\end{clm}

\begin{proof}
\mycommand{0}{\mathbf{0}}
Let $\0 = 0^n\in\cD$ and $\0' = 0^{n-1}\in\cD'$.
Then, using Eqs~\ref{eqn:Gamma}, \ref{eqn:ESelem} and~\ref{eqn:alpha}, we have
\begin{eqnarray*}
\delta_\cP^* \Gamma\delta_\cQ\* &=& \sum_{\blue{V}\subseteq[n]: |\blue{V}|=k} \sk[ {n\choose k}^{-1/2} e_{\0'}]^* G_\blue{V} e_{\0} \\
&=& \sum_{\blue{V}\subseteq[n]: |\blue{V}|=k} \sk[ {n\choose k}^{-1/2} e_{\0'}]^* \alpha_0 e_{\0'} \\
&=& \sum_{\blue{V}\subseteq[n]: |\blue{V}|=k} {n\choose k}^{-1/2} {n\choose k}^{-1/2} n^{k/(k+1)} 
\\[1ex]
&=&  n^{k/(k+1)}. 
\end{eqnarray*}
\par\mbox{}\\[-9ex]
\end{proof}

\begin{clm}
\label{clm:normGamma}
$\norm|\Gamma| \le \sB[{1+O\s[n^{-1/(k+1)}] }] n^{k/(k+1)}$.
\end{clm}

\begin{proof}
The norm of $\Gamma$ equals the square root of the norm of
\begin{equation}
\label{eqn:GammaGamma}
\Gamma^*\Gamma = \sum_{\blue{V}\subseteq [n]:|\blue{V}|=k} G_\blue{V}^*G_\blue{V}\*.
\end{equation}
We upper bound the latter by the maximal $\ell_1$-norm of a column of the matrix in Eq.~\ref{eqn:GammaGamma} in the Fourier basis.  
Fix $v\in\cD$.  From Eq.~\ref{eqn:GTeu}, we get the following estimate on the diagonal entry corresponding to $e_v$:
\[
e_v^*\Gamma^*\Gamma e_v\* \le {n\choose k} \alpha_{|v|}^2 \le n^{2k/(k+1)}.
\]
Now, let us estimate the off-diagonal entries in the column corresponding to $e_v$.  \blue{Equation}~\ref{eqn:GTeu} tells us that any off-diagonal entry can only come from $G_\blue{V}^*G_\blue{V}\*$ where $\blue{V}\cap\supp(v)\ne\emptyset$, and also each such $\blue{V}$ contributes at most $k$ off-diagonal entries.  Thus, the sum of the absolute values of \blue{these} off-diagonal entries is at most
\[
|v|{n\choose k-1}\cdot k\cdot \alpha_0^2 = O\s[n^{-1/(k+1)}] n^{2k/(k+1)}
\]
since we may assume that $|v|\le n^{k/(k+1)}$ (otherwise, $\alpha_{|v|}=0$).  Summing both contributions, we get the required bound. 
\end{proof}
\blue{This concludes the proof of Claim~\ref{clm:2average}.}
\end{proof}
\begin{proof}[Proof of Theorem~\ref{thm:1average}.]  
Let $\cS$ be the algorithm of Theorem~\ref{thm:1average}.  
We~apply Theorem~\ref{thm:adv} to the algo\-rithm $\cA$ defined as follows,
using the constraints from Claim~\ref{clm:2average} to evaluate $\Adv$.
First, $\cA$ executes $\cS$ on its input.  Let $\{t_1,\dots,t_k\}$ be the output of $\cS$.
The algorithm $\cA$ then queries the elements $x_{t_1},\dots,x_{t_k}$.  It accepts if $x_{t_1}+\cdots+x_{t_k} = w$, and rejects otherwise.

The query complexity of $\cA$ is $T+k = T+O(1)$.  
The acceptance probability on distribution $\cP$ is $s_\cP = \epsn$.
Also, since $\cA$ always rejects a negative input,
\[
s_\cQ \le \Pr_{x\sim \cQ} \skA[\text{the input $x$ is positive}] \le \frac1m{n\choose k},
\]
the last inequality following from the union bound.
Thus, we have the following estimate on $\tau(s_\cP,s_\cQ)$:
\[
\tau(s_\cP, s_\cQ) = \sqrt{\strut s_\cP s_\cQ} + \sqrt{\strut (1-s_\cP)(1-s_\cQ)}
\le \sqrt{\frac1m{n\choose k}} + 1 - \frac{\epsn}2,
\]
and using the conditions on $m$ and $\epsn$, we obtain:
\begin{eqnarray*}
\frac{\delta_\cP^*\Gamma\delta_\cQ\* -  \tau(s_\cP,s_\cQ) \|\Gamma\|} {\|\Gamma\hadamard\Delta_j \|}
&=& \frac{n^{k/(k+1)} - \sA[1-\Omega(\epsn)] \sB[{1+O\s[n^{-1/(k+1)}] }] n^{k/(k+1)} }{O(1)} \\
&=& \Omega\s[\epsn n^{k/(k+1)}].  
\end{eqnarray*}
\par\mbox{}\\[-8ex]\qed\mbox{}\\[-3ex]
\end{proof}

\section{\mbox{Composition Theorem for the Average-Case Adversary Bound}}
\label{sec:composition}

We now prove 
the last remaining theorem needed to obtain the 
 lower bound on the average case complexity of $\ksum \circ \psearch_\innern^\outern$
(see Section~\ref{sec:hidden}). Recall that in this version, each input variable $x_i\in\Gr_m$ is embedded into a ``bucket'', that is, a sequence
$(x_{i1},\dots,x_{i\ell})\in (\Gr_m\cup\{\star\})^\innern$ in which
exactly one element is non-$\star$.
To apply our average-case adversary lower bound method,
we need to define the probability distributions and the matrix that appears in Eq.~\ref{eqn:adv}
for the composed problem. Intuitively, this is done by tensoring the matrix of the two problems that are composed, as well as the vectors that represent the probability distributions.
However, defining the matrix correctly to get a lower bound for the composed problem requires a careful analysis.

\blue{We use} the distributions $\cP_\F$ and $\cQ_\F$ to pick 
inputs to the outer function $\F$, and the uniform distribution to place each element of the input
independently in its bucket. Formally, we write 
$\cP = \cP_\F \otimes U_\innern^{\otimes \outern}$,  
where $U_\innern$ is the uniform distribution over $[\innern\,]$ and the distributions are viewed as real-valued vectors indexed by elements of their supports. The definition of $\cQ$ is similar, starting from $\cQ_\F$.

\begin{lemma} 
\label{lem:composition}
Let $\F: A^\outern \rightarrow B$, 
$\psearch_\innern:P\rightarrow A$ where $P\subseteq (A\cup \{\star\})^\innern$ 
is the set of all possible buckets, $\HH=\F \comp \psearch_\innern^\outern$, and $\cP_\F$, $\cQ_\F$, $\cP$ and $\cQ$ defined as above.
Then for any real numbers $s_\cP$, $s_\cQ\in [0,1]$ and matrix $\Gamma_\F$, there exists a matrix $\Gamma_\HH$ such that
\[ \Adv(\Gamma_\HH;\cP,s_\cP; \cQ,s_\cQ) \geq \Adv(\Gamma_\F;\cP_\F, s_{\cP}; \cQ_\F, s_{\cQ} ) \, \sqrt{\innern-1} \, . \]
\end{lemma}

\begin{theorem}
\label{thm:bucket}
Any algorithm that finds a solution to the search version of $\ksum \comp\psearch_\innern^\outern$ within $T$ queries with probability $\epsn>0$  
on average over the uniform distribution on positive instances requires
\[
\frac T{\epsn} = \Omega\sB[\sqrt{\ell-1}\; n^{k/(k+1)} ]
\]
provided $m= \omega\Big(n^{k+\tfrac 2 {k+1}} \Big) $.
\end{theorem}

The rest of this section is devoted to the proof of Theorem~\ref{thm:bucket}.
It follows closely the proof of the composition theorem in \blue{Ref.}~\cite{bhkkls:crypto},
and in particular the adversary matrix for $\HH$ 
we use here has the same structure as the matrices considered in that paper.
This \blue{allows} us to re-use some of the calculations from that paper (see Claims~\ref{claim:comp2} and \ref{claim:comp3}).

\blue{We use} the following notation.
\blue{Let}~$\X,\Y \in A^\outern$ denote inputs to $\F$. Its components are $\X_i \in A$.
The value $\Gamma_\F[\X,\Y]$ is a scalar. Notice that for the $\ksum$ problem, the rows of the matrix
defined in the previous section are only defined for positive inputs. In order to reuse
the norm calculations \blue{from} 
the composition theorem in \blue{Ref.}~\cite{bhkkls:crypto}, we need to extend it to all possible inputs.
We~do so by extending the matrix for $\ksum$ with rows of zeros. This transform\blue{ation} does not
change the norm of the matrix.
Similarly, the vector $s_{\cP_\F}$ can be extended with zeros to be defined for any input.

\begin{proof}[Proof of Lemma~\ref{lem:composition}]
The adversary matrix for the composed problem $\HH$ \blue{is} denoted $\Gamma_\HH$.
\blue{We~consider} blocks of $\Gamma_\HH$ indexed by values $\X,\Y\!$, which we denote 
$\Gamma_\HH^{\X,\Y}$. 
(These $\innern^\outern \times \innern^\outern$ blocks are a 
submatrix corresponding to all the inputs for which the input 
to $\F$ is $\X$, in the rows, and~$\Y\!$, in the columns.)
As in \blue{Ref.}~\cite{bhkkls:crypto},
we define $\Gamma_\HH$ by blocks as follows:
\[ 
\Gamma_\HH^{\X,Y} = 
	\Gamma_F[\X,\Y] \cdot \bigotimes_{i \in [\outern]} \overline{\Gamma}^{\X_i,\Y_i},
\]
where for $a,b\in A$,
\[\overline{\Gamma}^{a,b} = 
\begin{cases}
	\|{\one_\innern - \I_\innern}\| \cdot \I_\innern	 &\text{if $a = b$}\\
	\one_\innern - \I_\innern & \text{otherwise}.
\end{cases}
\]
An optimal adversary matrix for \psearch\ can be obtained by taking $\one_\innern - \I_\innern$ for all blocks except the diagonal ones that are all zeroes.
But if we were using it, a block $\Gamma_\HH^{X,Y} $ would be zero whenever there is an $i$ such that $X_i=Y_i$. 
Using the matrix $\overline \Gamma$, with modified diagonal blocks, overcomes this issue.

From the distributions $\cP_F$ and $\cQ_F$, we define the
vector $\delta_{\cP_F} = \sqrt{\cP_F}$, that is, 
$\delta_{\cP_\F}[\X]=\sqrt{\Pr_{\X\sampledfrom \cP_F}[\X]}$ (similarly for $\delta_{\cQ_\F}$).
Again, we can split $\delta_{\cP_\F}$ into blocks $\delta_{\cP_\F}^\X$.

With these definitions in hand, we can compute the terms that appear in
Eq.~\ref{eqn:adv} of Definition~\ref{def:adv}. This is done in Claims~\ref{claim:comp1}, \ref{claim:comp2}, and \ref{claim:comp3}. When referring to Ref.~\cite{bhkkls:crypto},  we use $S_i = \one_\innern - \I\innern$ for all $i$ ($1\leq i \leq n$).

\begin{clm} \label{claim:comp1}
$\delta_\cP^\dagger \Gamma_\HH \delta_\cQ  = 
	\delta_{\cP_\F}^\dagger \Gamma_F \delta_{\cQ_F} \cdot \|{\one_\innern - \I_\innern}\|^\outern$.
\end{clm}

\begin{clm}\cite[claim on last line of page 409]{bhkkls:crypto} \label{claim:comp2}
$\|\Gamma_\HH\| = \|\Gamma_\F\| \cdot \|{\one_\innern - \I_\innern}\|^\outern$.
\end{clm}

\begin{clm}\cite[claim near the end of page 410]{bhkkls:crypto} \label{claim:comp3}
For a query $i$ that corresponds to index~$q$ in the bucket~$p$, 
$\|\Gamma_\HH\hadamard\Delta_i \| = \|\Gamma_\F \hadamard\Delta_p \| \cdot \| \one_\innern - \I_\innern\|^{\outern-1}\cdot \|({\one_\innern - \I_\innern)\hadamard \Delta_q}\|$.
\end{clm}

Claims~\ref{claim:comp2} and~\ref{claim:comp3} were proven in
\blue{the arXiv extended version of} Ref.~\cite{bhkkls:crypto}.
Although the claims in \blue{the original Crypto version of Ref.}~\cite{bhkkls:crypto} consider specifically the
Element Distinctness problem, the paper mentions that an explicit description of the adversary matrix is not needed (such a description was indeed unknown when this proof was given).
For this reason, these two claims apply to any outer function $\F$, and in particular to $\ksum$. 
\blue{Note} that \blue{the arXiv extended version of} Ref.~\cite{bhkkls:crypto} contains the proofs for arbitrary outer functions.

\begin{proof}[Proof of Claim~\ref{claim:comp1}]
\begin{eqnarray*}
\delta_\cP^\dagger \Gamma_\HH \delta_\cQ  &= &
	\sum_{\X,\Y} (\delta_\cP^\X)^\dagger \Gamma_\HH^{\X,\Y} \delta_\cQ^\Y  \\
	&=& 
	\sum_{\X,\Y} 
		(\delta_\cP^\X)^\dagger 
		\left(\Gamma_F[\X,\Y] \cdot \bigotimes_{i \in [\outern]} \overline{\Gamma}^{\X_i,\Y_i}\right)
  		\delta_\cQ^\Y  \\
	&=& 
	\sum_{\X,\Y} 
		(\delta_{\cP_F^\X} \otimes \sqrt{U_\innern^{\otimes \outern}})^\dagger 
		\left(\Gamma_F[\X,\Y] \cdot \bigotimes_{i \in [\outern]} \overline{\Gamma}^{\X_i,\Y_i}\right)
		(\delta_{\cP_F^\Y} \otimes \sqrt{U_\innern^{\otimes \outern}}) \\
	&=& 
	\sum_{\X,\Y} 
		\left(\delta_{\cP_F^\X}\right)^\dagger 
		\left(\Gamma_F[\X,\Y] \right)
		\left(\delta_{\cP_F^\Y} \right) 
		\sqrt{U_\innern^{\otimes \outern}}
		\bigotimes_{i \in [\outern]} \overline{\Gamma}^{\X_i,\Y_i}\sqrt{U_\innern^{\otimes \outern}}\\
	&=& 
	\sum_{\X,\Y} 
		(\delta_{\cP_\F^\X})^\dagger 
		\left(\Gamma_\F[\X,\Y] \right)
		(\delta_{\cP_F^\Y} ) 
		\prod_{i\in [n]}\|\overline{\Gamma}^{\X_i,\Y_i}\|\\
	&=& 
		\delta_{\cP_\F}^\dagger 
		\Gamma_\F 
		\delta_{\cP_\F} 
		\|\one_\innern - \I_\innern \| ^\outern
\end{eqnarray*}
which concludes the proof of the claim. 
\end{proof}

Using the fact that $\|\one_\innern - \I_\innern\| = \innern -1 $ and $\|({\one_\innern - \I_\innern)\hadamard \Delta_q}\| = \sqrt{\innern -1 }$ for any $q$, we immediately get Lemma~\ref{lem:composition} by substituting the values obtained in Claims~\ref{claim:comp1}, \ref{claim:comp2} and \ref{claim:comp3} 
into Definition~\ref{def:adv}. 
\end{proof}

\begin{proof}[Proof of Theorem~\ref{thm:bucket}]
Using the values computed in Section~\ref{sec:ksum} 
we get
\begin{eqnarray*}
T &=& \Omega \left({ \frac{\delta_{\cP_\F}^\dagger \Gamma_\F \delta_{\cP_\F} - \tau(s_\cP,s_\cQ) \|\Gamma_\F\|} {\|\Gamma_\F \hadamard \Delta_i\|}} \sqrt{\innern -1}\right)\\
&=& \Omega \left( n^{k/(k+1)} \sqrt{\innern -1} \left(\frac \epsn 2 -  \sqrt{\frac 1 m {n \choose k}} \,\right) \right)
\end{eqnarray*}

Suppose that $\epsn$ is non-vanishing. Since $m$ is chosen large enough to make ${\frac 1 m {n \choose k}}$ 
arbitrarily small, we get
\[ \frac T \epsn =\Omega{\left( \sqrt{l-1} n^{k/(k+1)} \right)}. \]
\par\mbox{}\\[-8ex]
\end{proof}

\section*{Acknowledgements}
\label{sec:acks}

\blue{We are grateful to Kassem Kalach, with whom this work has initiated many years ago.
Part of this work was performed when GB visited AB, then at \emph{QuSoft} in Amsterdam.

The work of AB is supported in part by the ERC Advanced Grant MQC.
The work of GB is supported in part by the Canadian Institute for Advanced Research (CIFAR), 
the Canada \mbox{Research} Chair program,
Canada's Natural Sciences and Engineering Research Council (NSERC)
and Qu\'ebec's Institut transdisciplinaire d'information quantique. 
The work of PH is supported in part by CIFAR and NSERC.
The work of MK is supported in part by EPSRC grant number EP1N003829/1
Verification of Quantum Technology.
The work of SL is supported in part by the European Union Seventh Framework
Programme (FP7/2007-2013) under grant agreement no.~600700 (QALGO)
and the French ANR Blanc grant RDAM ANR-12-BS02-005.
The work of LS is supported in part by NSERC discovery grant
and discovery accelerator supplements programs.
}


\begin{thebibliography}{10}

\bibitem{AaronsonShi}
S.~Aaronson and Y.~Shi.
\newblock Quantum lower bounds for the collision and the element distinctness
  problems.
\newblock {\em Journal of the ACM} \textbf{51}(4):595--605, 2004.

\bibitem{amb02}
A.~Ambainis.
\newblock Quantum lower bounds by quantum arguments.
\newblock {\em Journal of Computer and System Sciences} \textbf{64}:750--767, 2002.

\bibitem{amb03}
\blue{A.~Ambainis.
\newblock Polynomial degree and lower bounds in quantum complexity:
Collision and element distinctness with small range.
\newblock {\em Theory of Computing} \textbf{1}(1):37--46, 2005.}

\bibitem{BarMah09}
B.~Barak and M.~Mahmoody-Ghidary.
\newblock {M}erkle puzzles are optimal --- {A}n \mbox{$O(n^2)$--query} attack
  on any key exchange from a random oracle.
\newblock In {\em Advances in Cryptology -- Proceedings of Crypto~2009}, pages
  374--390, 2009.

\bibitem{belovs:phd}
A.~Belovs.
\newblock {\em Applications of the Adversary Method in Quantum Query
  Algorithms}.
\newblock PhD thesis, University of Latvia, 2014.

\bibitem{belovs:onThePower}
A.~Belovs and A.~Rosmanis.
\newblock On the power of non-adaptive learning graphs.
\newblock {\em Computational Complexity} \textbf{23}(2):323--354, 2014.

\bibitem{spalek:kSumLower}
A.~Belovs and R.~\v{S}palek.
\newblock Adversary lower bound for the {$k$-sum} problem.
\newblock In {\em Proceedings of 4th ACM Innovations in Theoretical Computer
  Science}, pages 323--328, 2013.

\bibitem{belovs:variations}
\blue{A.}~Belovs.
\newblock Variations on quantum adversary.
\newblock \url{http://arxiv.org/abs/1504.06943}, April 2015.

\bibitem{BB84}
C.\,H. Bennett and G.~Brassard.
\newblock Quantum cryptography: {P}ublic key distribution and coin tossing.
\newblock In {\em Proceedings of International Conference on Computers, Systems
  \& Signal Processing, Bangalore}, pages 175--179, 1984.
\newblock Republished in 30th Anniversary Commemorative Issue of {\em
  Theoretical Computer Science} \textbf{560}(Part~1):\mbox{7--11}, 2014.

\bibitem{BBHT}
M.~Boyer, G.~Brassard, P.~H{\o}yer and A.~Tapp.
\newblock Tight bounds on quantum searching.
\newblock {\em Fortschritte der Physik} \textbf{46}:493--505, 1998.

\bibitem{Harachov}
G.~Brassard.
\newblock Cryptography in a quantum world.
\newblock In {\em Proceedings of SOFSEM 2016: Theory and Practice of Computer Science}, pages
  3--16, 2016.

\bibitem{bhkkls:crypto}
G.~Brassard, P.~H{\o}yer, K.~Kalach, M.~Kaplan, S.~Laplante and L.~Salvail.
\newblock {M}erkle puzzles in a quantum world.
\newblock In {\em Advances in Cryptology -- Proceedings of Crypto~2011}, pages
  391--410, 2011.
\newblock Extended version available at \url{http://arxiv.org/abs/1108.2316v1}.

\bibitem{bhkkls:arxivv2}
G.~Brassard, P.~H{\o}yer, K.~Kalach, M.~Kaplan, S.~Laplante and L.~Salvail.
\newblock Key establishment {\`a} la {M}erkle in a quantum world.
\newblock \url{http://arxiv.org/abs/1108.2316v2}, February 2015.

\bibitem{ICQNM}
G.~Brassard and L.~Salvail.
\newblock Quantum {M}erkle puzzles.
\newblock {\em Proceedings of Second International Conference on Quantum, Nano,
  and Micro Technologies}, pages 76--79, 2008.

\bibitem{DH}
W.~Diffie and M.\,E. Hellman.
\newblock New directions in cryptography.
\newblock {\em IEEE Transactions on Information Theory} \textbf{22}(6):644--654, 1976.

\bibitem{grover}
L.\,K. Grover.
\newblock Quantum mechanics helps in searching for a needle in a haystack.
\newblock {\em Physical Review Letters} \textbf{79}(2):325--328, 1997.

\bibitem{hoyer:advNegative}
P.~H{\o}yer, T.~Lee and R.~\v{S}palek.
\newblock Negative weights make adversaries stronger.
\newblock In {\em Proceedings of 39th Annual ACM Symposium on Theory of Computing},
  pages 526--535, 2007.
\newblock \url{http://dx.doi.org/10.1145/1250790.1250867}
  {\path{doi:10.1145/1250790.1250867}}.

\bibitem{ImpagliazzoRudich}
\blue{R.}~Impagliazzo and \blue{S.}~Rudich.
\newblock Limits on the provable consequences of one-way permutations.
\newblock In {\em Proceedings of 21st Annual ACM Symposium on
  Theory of Computing}, pages 44--61, 1989.

\bibitem{Kutin}
\blue{S.~Kutin.
\newblock Quantum lower bound for the collision problem with small range.
\newblock {\em Theory of Computing} \textbf{1}(1):29--36, 2005.}

\bibitem{LMRSS11}
T.~Lee, R.~Mittal, B.\,W. Reichardt, R.~\v{S}palek and M.~Szegedy.
\newblock Quantum query complexity of state conversion.
\newblock In {\em Proceedings of 52nd Annual IEEE Symposium on Foundations of
  Computer Science}, pages 344--353, 2011.

\bibitem{Makarov}
\blue{L.}~Lydersen, \blue{C.}~Wiechers, \blue{C.}~Wittmann, \blue{D.}~Elser, \blue{J.}~Skaar and \blue{V.}~Makarov.
\newblock Hacking commercial quantum cryptography systems by tailored bright
  illumination.
\newblock {\em Nature Photonics} \textbf{4}(10):686--689, 2010.

\bibitem{MNRS}
F.~Magniez, A.~Nayak, J.~Roland and M.~Santha.
\newblock Search via quantum walk.
\newblock {\em SIAM Journal on Computing} \textbf{41}(1):142--164, 2011.

\bibitem{merkle74}
R.~Merkle.
\newblock Publishing a new idea.
\newblock \url{http://www.merkle.com/1974/}.

\bibitem{merkle78}
R.~Merkle.
\newblock Secure communications over insecure channels.
\newblock {\em Communications of the~ACM} \textbf{21}(4):294--299, 1978.

\bibitem{RSA}
R.\,L. Rivest, A.~Shamir and L.~Adleman.
\newblock A method for obtaining digital signatures and public-key
  cryptosystems.
\newblock {\em Communications of the~ACM} \textbf{21}(2):120--126, 1978.

\bibitem{Shor}
P.\,W. Shor.
\newblock Polynomial-time algorithms for prime factorization and discrete
  logarithms on a quantum computer.
\newblock {\em SIAM Journal on Computing} \textbf{26}:1484--1509, 1997.

\bibitem{Cocks}
P.~Wayner.
\newblock British document outlines early encryption discovery.\\
\newblock \url{http://www.nytimes.com/library/cyber/week/122497encrypt.html},
  New York Times Technology Cybertimes column, 24 December 1997.

\bibitem{HKL}
\blue{Y.}~Zhao, \blue{C.-H.\,F.} Fung, \blue{B.}~Qi, \blue{C.}~Chen and \blue{H.-K.}~Lo.
\newblock Quantum hacking: {E}xperimental demonstration of time-shift attack
  against practical quantum-key-distribution systems.
\newblock {\em Physical Review~A} \textbf{78}(4):042333, 2008.

\end{thebibliography}

\end{document}